\documentclass[12pt]{article}

\usepackage{times}
\usepackage{fullpage}
\usepackage{amsfonts,amssymb,amsmath,amsthm}
\usepackage{latexsym}
\usepackage{url}
\usepackage{thmtools}
\usepackage{thm-restate}
\usepackage{hyperref}
\usepackage[capitalise, nameinlink]{cleveref}
\usepackage{algorithm}
\usepackage{algorithmicx}
\usepackage[pdftex]{color}
\usepackage{graphicx}
\usepackage{enumitem}

\newcommand{\tO}{\widetilde{O}}
\def\Exp{\mathbb{E}}
\def\Var{\mathrm{Var}}

\usepackage{algorithm}
\usepackage[noend]{algpseudocode}
\DeclareMathOperator*{\badcycle}{\normalfont\texttt{bad-cycle}}
\newcommand{\Pcal}{\mathcal{P}}

\newtheorem{theorem}{Theorem}
\newtheorem{lemma}[theorem]{Lemma}

\newtheorem{claim}[theorem]{Claim}

\begin{document}
\title{Testing properties of signed graphs}
\author{Florian Adriaens\thanks{KTH, Sweden. Email: adriaens@kth.se} \and
Simon Apers\thanks{CWI, the Netherlands and ULB, Belgium. Email: smgapers@gmail.com}}
\date{}
\maketitle

\begin{abstract}
In graph property testing the task is to distinguish whether a graph satisfies a given property or is ``far'' from having that property, preferably with a sublinear query and time complexity.
In this work we initiate the study of property testing in signed graphs, where every edge has either a positive or a negative sign.
We show that there exist sublinear algorithms for testing three key properties of signed graphs: \emph{balance} (or \emph{2-clusterability}), \emph{clusterability} and \emph{signed triangle freeness}.
We consider both the dense graph model, where we can query the (signed) adjacency matrix of a signed graph, and the bounded-degree model, where we can query for the neighbors of a node and the sign of the connecting edge.

Our algorithms use a variety of tools from graph property testing, as well as reductions from one setting to the other.
Our main technical contribution is a sublinear algorithm for testing clusterability in the bounded-degree model.
This contrasts with the property of $k$-clusterability which is not testable with a sublinear number of queries.
The tester builds on the seminal work of Goldreich and Ron for testing bipartiteness.
\end{abstract}

\section{Introduction}
\label{sec:intro}

\paragraph*{Signed Graphs}

A \emph{signed graph} is a graph where every edge either has a positive or a negative label.
Formally it is denoted as $G = (V,E,\sigma)$ with node set $V = [N]$, edge set $E \subseteq V \times V$ and edge labelling $\sigma:E \to \{+,-\}$.
Such graphs model a variety of different scientific phenomena.
The widely studied correlation clustering problem \cite{Bansal, DEMAINE2006172} was orginally motivated by a document classification problem, where one has knowledge of pairwise similarities between documents, and the goal is to cluster the documents into (an undefined number of) groups such that within a group the documents are similar to each other, while across groups they are less similar.
Several authors \cite{Bonchi,nipsruochun} have focused on the related problem of finding large polarized communities.
In physics, signed graphs and their frustration index are utilized to model the ground-state energy of Ising models \cite{kasteleyn1963dimer}.
A third example are social networks, in which interactions between individuals can often be categorized into binary categories: trust versus distrust, friendly or antagonistic, etc.
An important aspect in the edge formation of social networks is the sign of triangles.
According to structural balance theory from social psychology \cite{Cartwright1956StructuralBA}, triangles with either one or three positive edges are more plausible, and this prevalence has been observed in real-life social networks \cite{LeskovecSigned,TangSurvey}.
Many methods and algorithms for link and sign prediction try to capitalize on this. 
Aside from link prediction, the survey of \cite{TangSurvey} lists several other important data mining tasks in signed social media networks.

Signed graphs generalize unsigned graphs, and as such they can have different properties than unsigned graphs.
One important example is the property of \emph{clusterability} or \emph{weak balance}, this was first introduced in \cite{davis1967clustering} and it is the subject of the correlation clustering problem \cite{Bansal}.
A signed graph is clusterable if there exists a partitioning of the nodes into an a priori unknown number of components such that (i) every positive edge connects two nodes in the same component, and (ii) every negative edge connects two nodes in different components.
An equivalent characterization, in terms of forbidden subgraphs, is that the signed graph contains no cycles with exactly one negative edge \cite{davis1967clustering, DEMAINE2006172}.
Clusterability does not appear to have a meaningful interpretation in the case of unsigned graphs.
For example, if one views an unsigned graph as a signed graph with the restriction that all the edges have the same sign, whether positive or negative, then clearly any such graph is clusterable since it contains no cycles with exactly one negative edge.

Other signed graph properties are closer related to unsigned graph properties.
For example, the property of \emph{balance} or \emph{2-clusterability} in signed graphs \cite{konig1936theorie,harary1953} in fact generalizes that of bipartiteness in unsigned graphs.
A signed graph is balanced if it is clusterable into exactly two components, with only positive edges inside the components and negative edges between the components.
It follows that a signed graph with only negative edges is balanced iff the underlying unsigned graph is bipartite.
There is also a reduction in the opposite direction, transforming a signed graph to an unsigned graph by replacing each positive edge by a path of two negative edges, and afterwards omitting all the signs of the edges \cite{ZASLAVSKY201831}.
This reduction preserves distances, in the sense that the minimum number of edges that need to be deleted to make the signed graph balanced (the \emph{signed frustration index}) is equal to the minimum number of edges that need to be deleted to make the transformed graph bipartite (the \emph{unsigned frustration index}) \cite[Proposition 2.2]{ZASLAVSKY201831}.
We will use this reduction to show that in the bounded-degree model we can reduce the problem of testing balance to that of testing bipartiteness.\footnote{This reduction does not work well in the dense graph model, since the transformed graph will be typically sparse.}

A final signed graph property that we investigate is that of \emph{signed triangle freeness}.
For a given signed triangle we say that a signed graph is signed triangle free if the given triangle does not occur in the graph.
The occurrence or absence of certain signed triangles is relevant in structural balance theory \cite{Cartwright1956StructuralBA} and it generalizes the notion of triangle freeness in unsigned graphs.

\paragraph*{Graph Property Testing}
Graph property testing was formally introduced in the seminal work of Goldreich, Goldwasser and Ron \cite{GoldreichLearning}.
As input we are given query access to an (unsigned) graph $G = (V,E)$ with node set $V = [N]$ and edge set $E \subseteq V \times V$.
We would like to decide whether the graph obeys a certain property $\Pcal$, or whether it is ``far'' from any graph having that property.
This is a relaxed setting as compared to that of deciding $\Pcal$ and it often allows for algorithms that have sublinear query and/or time complexity.
Such sublinear algorithms have been proposed for a wide range of graph properties such as bipartiteness \cite{GoldreichLearning,goldreich1999sublinear}, $k$-colorability \cite{GoldreichLearning,Alonkcol}, cycle-freeness \cite{GoldreichBounded,czumaj2014finding} and more generally monotone graph properties \cite{alon2008every} and minor-closed properties \cite{kumar2018finding,kumar2019random}.
We refer the interested reader to the survey by Goldreich \cite{goldreich2010introduction}.

The precise definition of ``far'' however depends on the type of query access that we have to the graph:
\begin{enumerate}
\item
In the \emph{dense graph model} \cite{goldreich_2017} we are able to query the adjacency matrix entries.
A query takes the form $(v,w) \in [N] \times [N]$ and the reply is $1$ if there is an edge between $v$ and $w$, otherwise it is $0$.
Two graphs $G=(V,E)$ and $G'=(V,E')$ are said to be $\epsilon$-far from each other if they differ in at least an $\epsilon$-fraction of the adjacency matrix entries.
Equivalently, at least $\epsilon N^2$ edges have to be added or removed to turn $G$ into $G'$.
\item
In the \emph{bounded-degree graph model} we are given an upper bound $d$ on the degrees of the graph, and we are given access to the adjacency list of $G$.
A query takes the form $(v,i)$, with $v \in [N]$ and $i \in [d]$.
If the degree of $v$ is at least $i$, then the query is answered with the $i^{\text{th}}$ neighbor $u$ of node $v$ (in arbitrary order).
If $v$ has degree smaller than $i$, an error symbol is returned.
Two graphs $G=(V,E)$ and $G'=(V,E')$ are $\epsilon$-far from each other if at least $\epsilon nd$ edges have to be modified (added or removed) to turn $G$ into $G'$.
\end{enumerate}
Using these definitions, a graph is \emph{$\epsilon$-far from having property $\Pcal$} if it is $\epsilon$-far from any graph having property $\Pcal$.

A property testing algorithm for $\Pcal$ is a randomized algorithm that, given query access to $G$ and an error parameter $\epsilon$, should behave as follows: (i) if $G$ has property $\Pcal$ then the algorithm should accept with probability at least $2/3$, whereas (ii) if $G$ is $\epsilon$-far from having property $\Pcal$, then the algorithm should reject with probability at least $2/3$.
If $G$ satisfies neither condition than the algorithm can behave arbitrarily.
This is the main reason why testing algorithms are often far more efficient than algorithms for effectively deciding whether $G$ has property $\Pcal$ or not.
If a property tester always accepts graphs having property $\Pcal$ (i.e., it never falsely rejects), then it is called a \emph{one-sided} property tester for $\Pcal$.
Otherwise it is called a \emph{two-sided} property tester.

\paragraph*{Testing in Signed Graphs}

Given the utility of graph property testing, and the importance of signed graphs in many applications, we believe that extending the framework of graph property testing to signed graphs is worthwhile.
The definitions of distance and query access for unsigned graphs are easily extended to signed graphs:
\begin{enumerate}
\item
In the \emph{dense signed graph model} adjacency matrix queries are now answered by an element from $\{0,-,+\}$.
A signed graph is $\epsilon$-far from property $\Pcal$ if a least $\epsilon N^2$ edge modifications (addition, removal or sign switch) have to be made to obtain a graph that satisfies $\Pcal$.
\item
In the \emph{bounded-degree signed graph model} a query $(v,i) \in [N] \times [d]$ is now answered either by the $i$-th neighbor $w$ of $v$ and the sign $\sigma(v,w)$ of the corresponding edge, or by an error symbol if $v$ has less than $i$ neighbors.
A signed graph is $\epsilon$-far from $\Pcal$ if a least $\epsilon nd$ edge modifications (addition, removal or sign switch) have to be made to obtain a graph that satisfies $\Pcal$.
\end{enumerate}
Note that in both cases the edge modifications now consist of edge additions, removals, as well as sign switches.
However, the properties discussed in this paper (signed triangle freeness, balance and clusterability) are all monotonous, and hence one may restrict the attention to edge removals.

\subsection{Results and Techniques}
In this work we investigate property testing algorithms for the canonical signed graph properties of signed triangle freeness, balance and clusterability.
Table~\ref{table:summaryresults} summarizes our results in terms of query complexity of the proposed testers.
The $\tO(\cdot)$-notation hides polylogarithmic factors in its argument and in $N$, in the bounded-degree model it also hides a polynomial dependence on the degree bound $d$.
The time complexity for testing balance and clusterability in the bounded-degree model is bounded by the query complexity $\tO(\sqrt{N}/\mathrm{poly}(\epsilon))$.
In all other cases the time complexity is at most exponential in the query complexity.
In the rest of the section we give a sketch of our techniques.

\begin{table}[ht]
\centering
\def\arraystretch{1.4}%
\begin{tabular}{|c|c|c|}
\hline
& Dense signed graph model & Bounded-degree signed graph model \\
\hline
Signed triangle freeness & $\tO(\mathrm{tower}(\log(1/\epsilon)))$ \cite{Fox} & $\tO(1/\epsilon)$ \cite{goldreich2010introduction} \\
\hline
Balance & $\tO(1/\epsilon)$ \cite{Sohlercannon} & $\tO(\sqrt{N}/\mathrm{poly}(\epsilon))$ \\
\hline
Clusterability & $\tO(1/\epsilon^7)$ & $\tO(\sqrt{N}/\mathrm{poly}(\epsilon))$ \\
\hline
\end{tabular}
\caption{Query complexity of the different property testers. All testers are one-sided except for the clusterability tester in the dense model. The function $\mathrm{tower}(\log(1/\epsilon)))$ denotes a power tower of $2$'s of height $O(\log(1/\epsilon))$.}
\label{table:summaryresults}
\end{table}

\paragraph*{Dense signed graph model}

We first describe property testing algorithms in the dense signed graph model.
The property of \emph{signed triangle freeness} can be efficiently tested by interpreting the signed graph as an edge-colored graph.
We can then use Fox's edge-colored triangle removal lemma \cite{Fox}, similar to the case of triangle freeness in unsigned graphs.
For the property of \emph{balance} or 2-clusterability we can use a reduction to a constraint satisfaction problem (CSP): every node corresponds to a Boolean variable (which indicates its cluster), a positive edge imposes an equality constraint between its endpoints and a negative edge imposes an inequality constraint.
We can then test balance by using a property testing algorithm for CSPs \cite{ALON2003212,Sohlercannon,andersson2002property}.
The property of \emph{clusterability} can also be cast as a CSP in which the node variables now take arbitrary integer values in $[N]$ (indicating their cluster).
However, the aforementioned CSP testers \cite{ALON2003212,Sohlercannon,andersson2002property} are not efficient in such a regime.
We circumvent this problem by proving that a signed graphs that is clusterable is necessarily $\epsilon/10$-close to being clusterable into $O(1/\epsilon)$ clusters.
Using this we reduce the problem of testing clusterability to that of distinguishing graphs that are $\epsilon/10$-close to being $O(1/\epsilon)$-clusterable from those that are $\epsilon$-far from being $O(1/\epsilon)$-clusterable.
This problem corresponds to \emph{tolerantly} testing a CSP where the variables now take values in $[O(1/\epsilon)]$, and this can be done efficiently using an algorithm by Andersson and Engebretsen \cite{andersson2002property}.

\paragraph*{Bounded-degree signed graph model}

Now we turn to the bounded-degree model.
Testing signed triangle freeness is trivial in this model, similar to the unsigned case \cite{goldreich2010introduction}.
Testing balance requires more care.
While we can again cast the problem as a CSP, we are not aware of any appropriate property testing algorithms for CSPs in the bounded-degree model.
Rather we reduce the problem of testing balance for signed graphs to that of testing bipartiteness for unsigned graphs, for which we can use the algorithm of Goldreich and Ron \cite{goldreich1999sublinear}.
The reduction is based on a transformation described by Zaslavsky \cite{ZASLAVSKY201831}, which maps balanced (resp.~unbalanced) signed graphs to bipartite (resp.~nonbipartite) unsigned graphs.
The resulting algorithm's query complexity has an optimal $\sqrt{N}$-dependence, which follows from the $\Omega(\sqrt{N})$ lower bound for testing bipartiteness.

Finally, and this is our main technical contribution, we describe a property testing algorithm for clusterability.
While we can again reduce the problem to testing $O(1/\epsilon)$-clusterability, similar to the dense case, the problem is that $k$-colorability (which is a special case of $k$-clusterability) is not testable in the bounded-degree model \cite{1313803}.
Rather, we base our algorithm on the forbidden subgraph characterization by Davis \cite{davis1967clustering}, which states that a signed graph is clusterable if and only if it has no cycles with exactly one negative edge.
We then use random walks to find such a cycle: first we pick a random initial node and perform a large number of random walks \emph{on the positive edges} of $G$, then we check for the existence of a negative edge between any pair of nodes that were visited by a random walk.
Such a negative edge necessarily yields a bad cycle.
The correctness of this algorithm is easy to prove when the positive edges in $G$ induce an expander.
For the general case we build on the (unsigned) graph decomposition results of Goldreich and Ron \cite{goldreich1999sublinear}.

\subsection{Conclusion and open questions}
On the one hand, our work demonstrates that key properties of signed graphs can be tested very efficiently.
This seems to not have been studied before.
On the other hand, we introduce signed graphs as an interesting setting for studying graph property testing.
Our work leaves open a number of questions and future directions:
\begin{itemize}
\item
A lot of effort has been put in characterizing the set of properties that are testable (using $\tO(\mathrm{poly}(1/\epsilon))$ queries) in the dense graph model \cite{alon2008every,AlonTestable,GoldReich2,alon2009combinatorial} and the bounded-degree graph model \cite{benjamini2010every,czumaj2009testing,czumaj2014finding,newman2013every}.
It would be interesting to characterize the set of signed graph properties that are testable.
\item
In a very recent work by Kumar, Seshadhri and Stolman \cite{kumar2021random} an efficient \emph{partition oracle} was proposed.
For minor-closed graph families such an oracle gives local access to a certain global decomposition of the graph.
The study of such a decomposition and corresponding oracle for signed graphs seems like an interesting future direction, especially given the connection between signed graphs and social networks.
\item
Finally, we did not succeed in proving a $\Omega(\sqrt{N})$ lower bound for testing clusterability in the bounded-degree signed graph model, and hence we leave this as an open question.
\end{itemize}

\section{Property testing in the dense signed graph model}

\subsection{Signed Triangle Freeness}
A signed triangle is any triangle with a fixed sign assignment of its edges.
As mentioned in the introduction, it is often interesting to check whether a signed graph for instance contains any triangles with exactly one negative edge.
While it can be computationally expensive to effectively decide this, especially for massive graphs such as social networks, it might be easier to \emph{test} whether the graph is free of such triangles.

Testing triangle freeness for unsigned graphs has been well studied.
It is a direct application of the triangle removal lemma \cite{ruzsa1978triple,Fox}.
For some function $f$, the canonical one-sided tester simply picks $f(\epsilon)$ random triples in $[N]$ and rejects the graph if any of the triples induces a triangle.
If the graph is triangle free, then we always accept the graph, so that the tester is indeed one-sided.
If the graph is $\epsilon$-far from being triangle free, than the triangle removal lemma of Fox \cite{Fox} proves that the graph contains at least $\delta(\epsilon) \, \binom{n}{3}$ triangles.
Here $\delta(\epsilon)$ is a function bounded by the inverse of the towering function $\mathrm{tower}(\log(1/\epsilon))$, which corresponds to a tower of $2$'s of height $O(\log(1/\epsilon))$ (i.e., $2$-to-the-$2$-to-the-\dots-to-the-$2$, $O(\log(1/\epsilon)$ times).
Hence if we sample $f(\epsilon) \in \Theta(1/\delta(\epsilon))$ triples, then with high probability one of them will induce a triangle and we will correctly reject the graph.
This yields a total query complexity of $O(1/\delta(\epsilon))$.

Testing signed triangle freeness in signed graphs can be analyzed in a very similar way.
By interpreting the edge signs of a graph as an edge-coloring, we can use Fox's \emph{colored} triangle removal lemma \cite{Fox}.
This lemma states that if an edge-colored graph is $\epsilon$-far from being free of a certain colored triangle, then the graph contains at least $\delta'(\epsilon) \, \binom{n}{3}$ such induced triangles, where $\delta'(\epsilon)$ is again bounded by the inverse of $\mathrm{tower}(\log(1/\epsilon))$.
By the same argument as in the unsigned case, this implies the existence of a one-sided tester for colored (or signed) triangle freeness in the dense graph model with query complexity $O(\mathrm{tower}(\log(1/\epsilon)))$.
This proves the following theorem.
\begin{theorem}
There exists a one-sided tester for signed triangle freeness in the dense signed graph model with query complexity $\tO(\mathrm{tower}(\log(1/\epsilon)))$.
\end{theorem}

\subsection{Balance}
We can cast balance or 2-clusterability of a signed graph $G = (V,E,\sigma)$ as a satisfiability problem.
Associate with each node $v$ a variable $x_v \in \{0,1\}$.
With every edge $(u,v) \in E$ we associate a constraint on $x_u$ and $x_v$: if $\sigma(e) = +$ (positive edge) the constraint is satisfied iff $x_u = x_v$; if $\sigma(e) = -$ (negative edge) then the constraint is satisfied iff $x_u \neq x_v$.
The graph $G$ will be balanced iff there exists an assignment of $x_v$'s such that all constraints are satisfied.
Even more, if $G$ is $\epsilon$-far from being balanced then we similarly have to remove $\epsilon n^2$ constraints from the satisfiability problem for it to satisfiable.
As a consequence, the problem reduces to testing whether the satisfiability problem is in fact satisfiable.
For this we can use the work by Sohler \cite{Sohlercannon} which describes a one-sided tester with query complexity $\tO(1/\epsilon)$.
The algorithm is very simple: sample $\tO(1/\epsilon)$ variables and accept, if and only if the induced set of constraints on those variables has a satisfying assignment.
Applying this algorithm to the problem of testing balance gives the following algorithm: sample $\tO(1/\epsilon)$ nodes, query the entire induced subgraph, and accept if and only if the induced subgraph is balanced.
From \cite[Theorem 1]{Sohlercannon} it then follows that this describes a one-sided property tester for balance.
This proves the following theorem.
\begin{theorem}
There exists a one-sided tester for balance in the dense signed graph model with query complexity $\tO(1/\epsilon)$.
\end{theorem}

We note that testing $k$-clusterability can be reduced to satisfiability in the very same manner, except that now the variables $x_v \in \{0,1,\dots,k-1\}$.
For constant $k$ we can again use \cite[Theorem 1]{Sohlercannon} to get a one-sided tester with query complexity $\tO(1/\epsilon)$.

\subsection{Clusterability}
A relaxation of $k$-clusterability for signed graphs is the notion of \emph{weak balance} or \emph{clusterability} \cite{davis1967clustering,Bansal}. 
A signed graph is clusterable if it is $k$-clusterable for some (a priori unknown) $k \in [N]$.
Since there can only be $N$ clusters, we could test clusterability by testing $N$-clusterability.
The satisfiability reduction from last section however fails in such case, because typical satisfiability testers have a bad dependence on the domain size of the variables.

Rather, we argue that testing clusterability can be reduced to \textit{tolerantly} testing $k$-clusterability for $k \in O(1/\epsilon)$.
A tolerant tester \cite{PARNAS20061012} is required to accept inputs that are $\epsilon_1$-close to some property $\mathcal{P}$, while rejecting inputs that are  $\epsilon_2$-far from $\mathcal{P}$, for some parameters $\epsilon_1<\epsilon_2$.
Tolerant testing is closely related to approximating the distance from an object to a property (see \cite{PARNAS20061012}).
We use the following lemma.

\begin{lemma} \label{lem:cluster}
If a signed graph is clusterable then it is $4\epsilon$-close to being clusterable into at most $1/\epsilon$ clusters.
\end{lemma}
\begin{proof}
Let the partition $V = P_1 \cup P_2 \cup \dots \cup P_r$ denote a valid clustering of the graph.
We define a new partition by merging different components: keep all components $P_i$ of size $|P_i| \geq \epsilon N$, and merge the remaining components into components of size between $\epsilon N$ and $2 \epsilon N$ (which is always possible).
This yields a new partition with at most $1/\epsilon$ components.
Between the components there are only negative edges, and there are at most $(2\epsilon N)^2/\epsilon = 4 \epsilon N^2$ edges within the components.
Hence if we remove all the edges within the new components, then we obtain a new graph for which the new partition describes a clustering with at most $1/\epsilon$ clusters, and which is $(4\epsilon)$-far from the original graph.
\end{proof}

Now if a signed graph is $\epsilon$-far from being clusterable, then clearly it is also $\epsilon$-far from being say $(8/\epsilon)$-clusterable.
On the other hand, by this lemma, a graph that is $\epsilon/4$-close to being clusterable will be $(\epsilon/4 + \epsilon/2)=3\epsilon/4$-close to a graph that is $(8/\epsilon)$-clusterable.
Hence we can use a \textit{tolerant tester} for $O(1/\epsilon)$-clusterability to tolerantly test clusterability.
Equivalently, we can use an additive estimate (with error $\pm\epsilon N^2$) on the number of edges that need to be removed in order to make a signed graph $k$-clusterable, for $k \in O(1/\epsilon)$.
Now we are in better shape to cast the problem as a satisfiability problem, similar to last section.
In \cref{app:dense-cluster} we detail how to use the algorithm of Andersson and Engebretsen \cite{andersson2002property} to tolerantly test for $O(1/\epsilon)$-clusterability using $\tO(1/\epsilon^7)$ queries.
This yields the following theorem.
\begin{theorem}
There exists a two-sided tolerant tester for clusterability in the dense signed graph model with query complexity $\tO(1/\epsilon^7)$.
\end{theorem}
\noindent
Since the tester is tolerant, this also gives an algorithm to estimate the weak frustration index with additive error $\epsilon N^2$ using $\tO(1/\epsilon^7)$ queries.

%%%%%%%%%%%%%%%%%%%%%%%%%%%%%%%%%%%%%%%%%%%%%%%%%%%%
\section{Property testing in the bounded-degree signed graph model}
%%%%%%%%%%%%%%%%%%%%%%%%%%%%%%%%%%%%%%%%%%%%%%%%%%%%

\subsection{Signed triangle freeness}

Testing triangle freeness in the bounded-degree (unsigned) graph model is significantly easier to analyze than doing so for the dense model \cite{goldreich2010introduction}.
In fact, the exact same argument applies to signed graphs and we will describe it here for completeness.
Given query access to a signed graph, consider the following testing algorithm: pick $\tO(1/\epsilon)$ nodes and reject if any of them is part of a signed triangle.
We can check whether a node is part of a signed triangle simply by querying for its neighbors, and the neighbors of its neighbors.
In the bounded-degree model this takes only $\tO(1)$ queries.
If the graph is signed triangle free then we will never reject.
On the other hand, if at least $\epsilon d N$ edges have to removed in order to make the graph signed triangle free, then at least $\epsilon N$ nodes must be part of a signed triangle.
With large probability the algorithm will sample such a node and consequently reject the graph.
This yields the theorem below.
\begin{theorem}
There exists a one-sided tester for signed triangle freeness in the bounded-degree signed graph model with query complexity $O(1/\epsilon)$.
\end{theorem}

%%%%%%%%%%%%%%%%%%%%%%%%%%%%%%%%%%%%%%%%%%%%%%%%%%%%
\subsection{Balance}
%%%%%%%%%%%%%%%%%%%%%%%%%%%%%%%%%%%%%%%%%%%%%%%%%%%%

Our algorithm for testing balance of bounded-degree signed graphs reduces the problem to testing bipartiteness in a related unsigned graph\footnote{The reduction is reminiscent of the reduction from cycle-freeness testing to bipartiteness testing in \cite{czumaj2014finding}. However, the reduction in \cite{czumaj2014finding} is randomized and between unsigned graphs, whereas our reduction is deterministic and for signed graphs.}.
Consider the following mapping from a signed graph $G$ to an unsigned graph $G'$: (i) for every positive edge $(u,v)$ create a new node $w^{(u,v)}$ and replace the edge $(u,v)$ by two unsigned edges $(u,w^{(u,v)})$ and $(w^{(u,v)},v)$, and (ii) replace each of the remaining negative edges by an unsigned edge.
The unsigned graph $G'$ has an odd cycle if and only if $G$ has a cycle with an odd number of negative edges.
As a consequence, $G'$ will be bipartite if and only if $G$ was balanced.

\begin{figure}[htb]
\centering
\includegraphics[width=.85\textwidth]{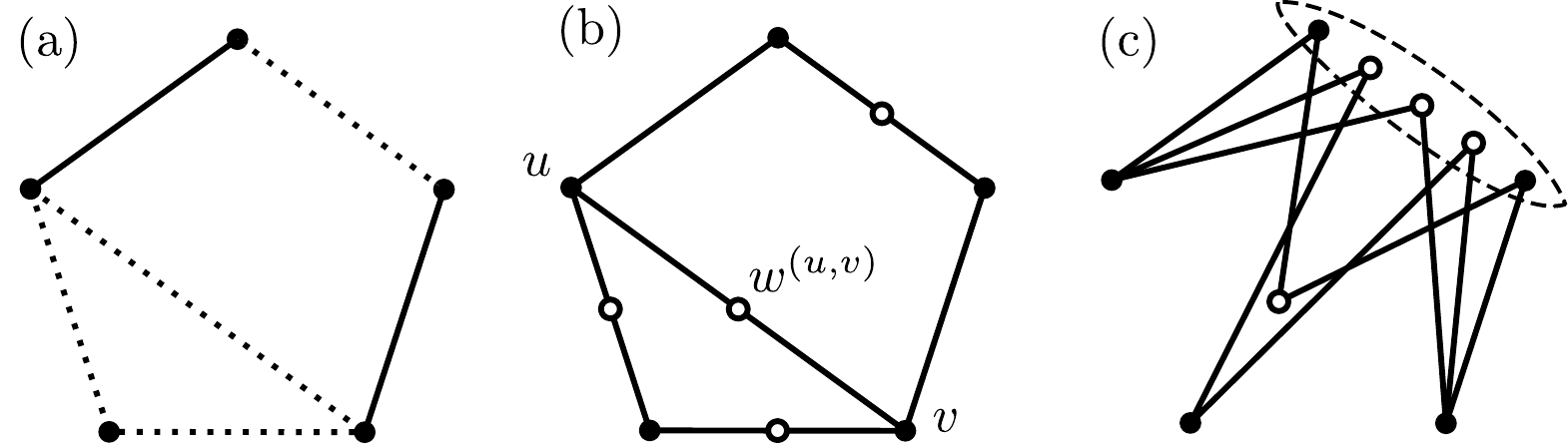}
\caption{(a) Balanced signed graph $G$ (dotted lines are positive edges, solid lines are negative edges). (b) Bipartite unsigned graph $G'$ after mapping. (c) Bipartition of $G'$.}
\label{fig:balance-reduction}
\end{figure}

In fact, an even stronger property holds \cite[Proposition 2.2]{ZASLAVSKY201831}: the (signed) frustration index of $G$ is equal to the (unsigned) frustration index of $G'$.
This implies the following lemma.
\begin{lemma}
If $G$ is $\epsilon$-far from balanced then $G'$ is $\epsilon/(d+1)$-far from bipartite.
\end{lemma}
\begin{proof}
For the second fact, let $G$ be $\epsilon$-far from being balanced, so that it has signed frustration index $k \geq \epsilon d N$.
The unsigned graph $G'$ then has unsigned frustration index $k \geq \epsilon d N$.
Now if $G$ has $m^+$ positive edges, then $G'$ has exactly $N+m^+ \leq (d+1) N$ vertices while keeping the same degree bound $d$.
As a consequence, we can bound its frustration index $k \geq \frac{\epsilon}{d+1} (d+1) d N \geq \frac{\epsilon}{d+1} d (N+m^+)$, so that $G'$ is indeed $\frac{\epsilon}{d+1}$-far from being bipartite.
\end{proof}

This lemma proves that we can test $\epsilon$-balancedness of $G$ by testing $\epsilon/(d+1)$-bipartiteness of $G'$.
For this we can use the following algorithm, which uses random walks to find odd cycles.
It was proven to be a one-sided bipartiteness tester in the bounded-degree model by Goldreich and Ron \cite{goldreich1999sublinear}.

\begin{algorithm}[H]
\caption{Bipartiteness tester} \label{alg:bounded-bipartite}
\begin{algorithmic}[1]
\For{$O(1/\epsilon)$ times}
\State
Pick a node $v$ uniformly at random.
\State
Perform $\tO(\sqrt{N}/\epsilon^3)$ random walks\footnotemark{} starting from $v$, each of length $\tO(1/\epsilon^8)$.
\State
If some vertex $u$ is reached both after an even path and after an odd path then reject.
\EndFor
\end{algorithmic}
\end{algorithm}
\footnotetext{A single step of the random walk from a node $v$ (with degree $d(v)$) corresponds to the following process: with probability $d(v)/d$ move to a uniformly random neighbor, and otherwise stay at $v$.}

\noindent
It remains to prove that we can efficiently implement this tester on $G'$.

\begin{lemma}
It is possible to implement \cref{alg:bounded-bipartite} on $G'$ using $\mathrm{poly}(\log(N)/\epsilon) \sqrt{N}$ adjacency list queries to $G$.
\end{lemma}
\begin{proof}
We need to be able to select a uniformly random node from $G'$, and implement a random walk on $G'$.
The latter is easy:
\begin{itemize}
\item
If we are on an original node $u$ in $G'$ then pick a random neighbor $v$ of $u$ in $G$.
If $(u,v)$ is negative, go to $v$, otherwise go to the new node indexed $w^{(u,v)}$.
\item
If we are on a new node $w^{(u,v)}$, go to either $u$ or $v$ with probability $1/2$.
\end{itemize}
To select a uniformly random node from $G'$, do the following:
\begin{enumerate}
\item
Pick $(u,i) \in [N] \times [d]$ uniformly at random and query for the $i$-th neighbor $v$ of $u$ in $G$.
If $u$ has less than $i$ neighbors we reject.
% Otherwise we found an edge $(v,w)$ with probability $1/(dN)$.
\item
If $\sigma(u,v) = -$, with probability $1/(4d(u))$ output a random endpoint of $(u,v)$ and terminate.
Otherwise, go to next step.
\item
If $\sigma(u,v) = +$, with probability $1/4$ output $w^{(u,v)}$ and terminate.
Otherwise, with probability $1/(3d(u))$ output a random endpoint of $(u,v)$.
\end{enumerate}
With probability $\geq 1/(4d)$, this scheme returns a uniformly random node from $G'$ (and otherwise it rejects).
To see this, first consider any original node $u$ in $G'$.
Any of its $d(u)$ incident edges is picked with an equal probability $1/(dN)$.
If a negative incident edge is picked, then $u$ is returned in step 2.~with probability $1/(4d(u))$; if it is a positive incident edge then $u$ is returned in step 3.~with probability $(1-1/4)/(3 d(u)) = 1/(4d(u))$.
Hence the total probability that $u$ is returned is
\[
d(u) \frac{1}{4 d(u)} \frac{1}{dN}
= \frac{1}{4 dN}.
\]
Now consider a new node $w^{(u,v)}$ in $G'$.
In step 1.~the edge $(u,v)$ is picked with probability $1/(dN)$, after which $w^{(u,v)}$ is returned with probability $1/4$ in step 3., yielding a total probability $1/(4 dN)$.
Since there are $N + m^- \geq N$ nodes in $G'$, the total probability of returning a node is $\geq N/(4 dN) = 1/(4d)$.
The sampling scheme only requires a single query, and so we can sample a uniformly random node from $G'$ using $4d \in O(1)$ queries in expectation.
By Chebyshev's inequality the total number of queries will be close to its expection with overwhelming probability.
\end{proof}

\noindent
This proves the following theorem.
\begin{theorem} \label{thm:bounded-balance}
There exists a one-sided tester for balance in the bounded-degree signed graph model with query complexity $\tO(\sqrt{N}/\mathrm{poly}(\epsilon))$.
\end{theorem}

Since balancedness of signed graphs generalizes bipartiteness of unsigned graphs, the $\widetilde\Omega(\sqrt{N})$ lower bound for testing bipartiteness in the (unsigned) bounded-degree model \cite[Theorem 7.1]{GoldreichBounded} also applies to testing balancedness in the (signed) bounded-degree model.
As a consequence, the $\sqrt{N}$-dependency of our tester is optimal.

%%%%%%%%%%%%%%%%%%%%%%%%%%%%%%%%%%%%%%%%%%%%%%%%%%%%
\subsection{Clusterability}
%%%%%%%%%%%%%%%%%%%%%%%%%%%%%%%%%%%%%%%%%%%%%%%%%%%%

In this section we prove the existence of a one-sided property tester for clusterability in the bounded-degree signed graph model.
We first note that similar to the dense case we can reduce the problem to testing $O(1/\epsilon)$-clusterability.
However, $k$-clusterability is a special case of $k$-colorability for unsigned graphs, and this is known not to be testable in the bounded-degree model \cite{bogdanov2002lower} (i.e., it requires $\Omega(N)$ queries).
Instead, we use the forbidden subgraph characterization of clusterability by Davis \cite{davis1967clustering}:
\begin{theorem}[{\cite[Theorem 1]{davis1967clustering}}] \label{thm:davis}
A signed graph $G$ is clusterable if and only if $G$ contains no cycle with exactly one negative edge.
\end{theorem}
We will call such a cycle a \emph{bad cycle}.
This characterization is a crucial distinction between clusterability and the untestable $k$-clusterability (or its unsigned variant, $k$-colorability), which does not seem to have such a simple characterization.

Similar to the bipartiteness tester of Goldreich and Ron \cite{goldreich1999sublinear} we will try to find bad cycles by performing many random walks in $G$.
Specifically, we simulate random walks on the unsigned subgraph $G^+ = (V,E^+)$ induced by the positive edges $E^+ = \{e \in E \mid \sigma(e) = + \}$.
Starting from a random initial node, we perform many such random walks and we check for the existence of a negative edge between distinct random walks.
Such a negative edge will necessarily yield a bad cycle, in which case we can safely reject the graph.

\begin{algorithm}[ht]
\caption{Tester for clusterability} \label{alg:weak-balance}
\begin{algorithmic}[1]
\For{$O(1/\epsilon)$ times}
\State
Pick a random node $s$ and run $\badcycle(s)$.
\State
If this returns a bad cycle, reject the graph.
\EndFor
\end{algorithmic}

$\badcycle(s)$:
\begin{algorithmic}[1]
\State
Perform $\tO(\sqrt{N}/\mathrm{poly}(\epsilon))$ random walks of length $\tO(1/\mathrm{poly}(\epsilon))$ on $G^+$, starting from $s$.
Let $K$ denote the set of all the nodes that are visited.
\State
If there is a negative edge between any pair of nodes in $K$, return the corresponding bad cycle.
\end{algorithmic}
\end{algorithm}

\noindent
We prove the following claim.

\begin{restatable}{theorem}{clusterability}
Algorithm \ref{alg:weak-balance} is a one-sided tester for clusterability with query complexity
\[
\tO(\sqrt{N}/\mathrm{poly}(\epsilon)).
\]
\end{restatable}

The claim about the query complexity is easy to check.
The total number of random walk steps is $\tO(\sqrt{N}/\mathrm{poly}(\epsilon))$, and a single random walk step can be implemented with $\tO(1)$ queries.
To check whether there exists a negative edge between any pair of nodes in $K$, it suffices to query the full (bounded) neighborhood of every node in $K$.
This takes $d |K| \in \tO(\sqrt{N}/\mathrm{poly}(\epsilon))$ queries.
The remainder of this section is used to prove correctness of the tester (which ultimately follows from \cref{claim:eps-good}).

\subsubsection{Intuition: Expanding case}
We will first describe the intuition behind the tester.
To this end, assume that there is a decomposition $V = V_1 \cup \dots \cup V_k$ as in \cref{fig:weak-balance} such that for each $i$ the following holds:
\begin{enumerate}
\item
$V_i$ has few positive outgoing edges:
\[
|E^+(V_i,V_i^c)|
\leq \frac{\epsilon}{2} d |V_i|.
\]
\item
A random walk of length $\tO(1/\mathrm{poly}(\epsilon))$ on $G^+$, and starting from any $s \in V_i$, ends uniformly at random inside $V_i$.
\end{enumerate}
While such a decomposition does not generally exist, the existence of a closely related decomposition was proven by Goldreich and Ron \cite{goldreich1999sublinear}.

Now assume that $G$ is $\epsilon$-far from being clusterable.
Then we claim that there must be at least $\frac{\epsilon}{2} d N$ negative edges inside the partitions $V_1,\dots,V_k$.
Indeed, if this were not the case, then we could find a valid clustering by removing these $\leq \frac{\epsilon}{2} d N$ negative edges together with the $\leq \frac{\epsilon}{2} d N$ positive edges between the partitions.
This contradicts the fact that $G$ is $\epsilon$-far from being clusterable.

Now make the additional assumption that the number of negative edges $|E^-(V_i)|$ inside each partition $V_i$ is $\Omega(\epsilon |V_i|)$, and consider an arbitrary node $s \in V_i$.
The probability that a pair of random walks on the positive edges, starting from $s$, results in a bad cycle can be lower bounded by the probability that the random walk endpoints $u$ and $v$ form a negative edge $(u,v) \in E^-$.
Since $u$ and $v$ are distributed uniformly, this probability is at least $|E^-(V_i)|/|V_i|^2 \in \Omega(\epsilon/N)$.
Taking $\sqrt{N/\epsilon}$ independent random walks (and ignoring correlations), the total probability of finding a bad cycle then becomes $\Omega\left(\binom{\sqrt{N/\epsilon}}{2} \frac{\epsilon}{N}\right) \in \Omega(1)$.

\begin{figure}[htb]
\centering
\includegraphics[width=.65\textwidth]{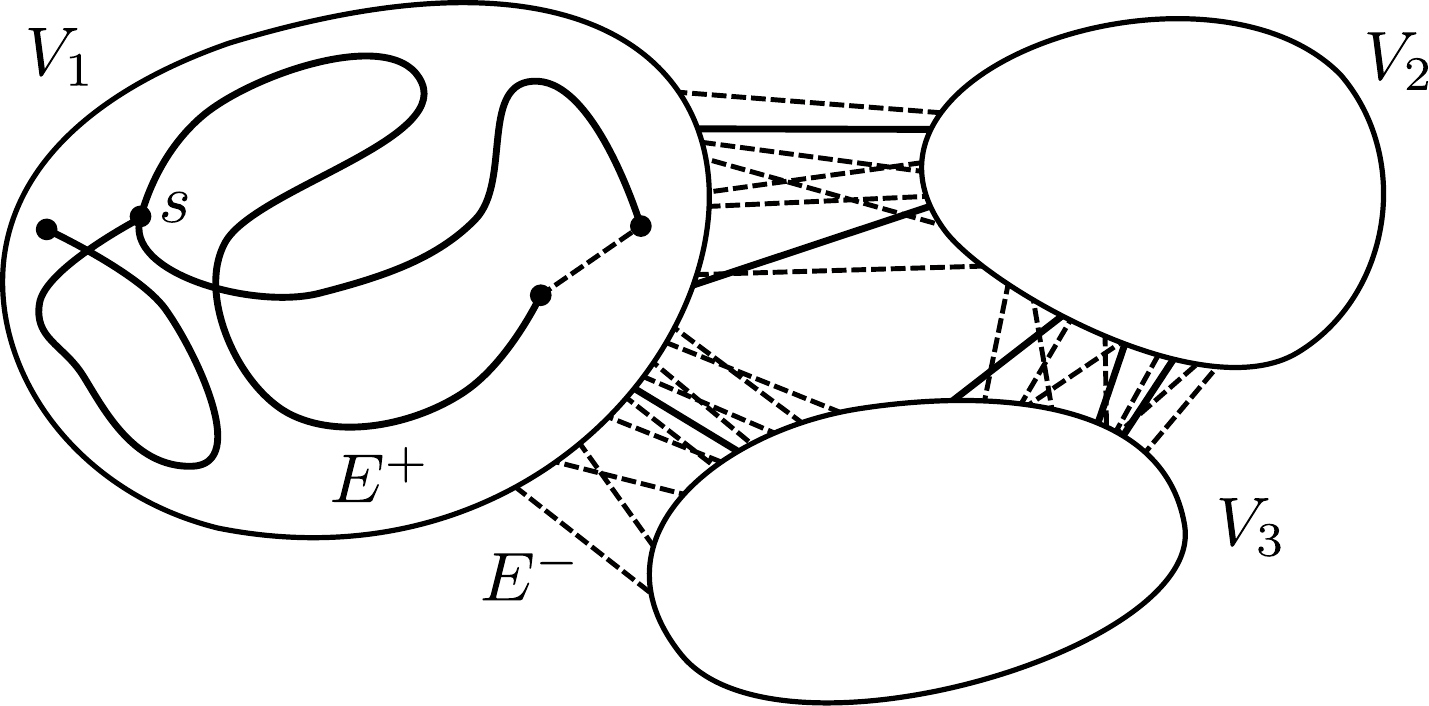}
\caption{Bounded degree tester for clusterability. Positive edges $E^+$ depicted as solid lines, negative edges $E^-$ depicted as dashed lines. $G$ is $\epsilon$-close to clusterable if there are $\leq \frac{\epsilon}{2} d N$ positive edges between the partitions and $\leq \frac{\epsilon}{2} d N$ negative edges inside the partitions.}
\label{fig:weak-balance}
\end{figure}

\subsubsection{General case: GR decomposition}
While the former section correctly captures the intuition behind the tester, the full proof of correctness is significantly more involved.
Luckily, most of the analysis runs similar to that of Goldreich and Ron \cite{goldreich1999sublinear} (regarding the graph decomposition, but also regarding bounds on the correlation between distinct random walks).
We can make the above intuition rigorous.
In the following let $G^+ = (V,E^+)$ denote the (unsigned) subgraph induced by the positive edges.
The main idea of the decomposition is that for most of the vertices $s$ in $G^+$, we can find a subset $S$ so that $S$ has few outgoing edges and a short random walk from $s$ mixes approximately uniformly over $S$.
We can hence set the first partition $V_1 = S$.
Now we would like to repeat the argument for the remaining graph $H^+ = G^+[V - V_1]$, induced on the node subset $V - V_1$.
The problem is that the next subset $S' \subseteq V - V_1$ will be ``good'' for random walks in $H^+$, but these can behave very differently from the original random walks in $G^+$.
This problem is dealt with by defining a Markov chain $M(H^+)$ on the unpartitioned subgraph $H^+$ such that (i) we can use $M(H^+)$ to cut off a new partition $S'$ from $H^+$, but also (ii) that the behavior of walks according to $M(H^+)$ is related to the behavior of the original random walks in $G^+$.
Details of the Markov chain $M(H^+)$ are given in \cref{app:mod-MC}.
The following lemma is proven in \cite{goldreich1999sublinear}, with $q_{s,v}(t)$ denoting the probability that $t$ steps of the Markov chain $M(H^+)$, starting from $s$, end in $v$.

\begin{lemma}[{\cite[Corollary 3 and Lemma 4.3]{goldreich1999sublinear}}] \label{lem:decomp}
Let $H^+$ be a subgraph of $G^+$ with at least $\epsilon N/4$ vertices.
Then for at least half of the vertices $s$ in $H^+$ there exists a subset of vertices $S$ in $H^+$, a value $\beta \in \widetilde\Omega(\epsilon^2)$ and an integer $t \in \tO(1/\epsilon^3)$ such that
\begin{enumerate}
\item
The number of edges between $S$ and the rest of $H^+$ is at most $\epsilon d |S|/2$.
\item
For every $v \in S$ it holds that $\sqrt{\frac{\beta}{|S| |H^+|}} \leq q_{s,v}(t) \leq \frac{1}{\epsilon} \sqrt{\frac{\beta}{|S| |H^+|}}$.
\end{enumerate}
\end{lemma}

This is the key lemma that underlies the GR decomposition.
Under these conditions, we can prove that if the original signed graph $G$ has many negative edges inside the subset $S$ then the modified Markov chain $M(H^+)$ will find a bad cycle.
This lemma is new, but its proof (which we defer to \cref{app:bad-cycle}) runs along the lines of the proof of \cite[Lemma 4.5]{goldreich1999sublinear}.
\begin{restatable}{lemma}{badcyclelemma}
Let $H^+$ be a subgraph of $G^+$, $s$ a vertex in $H^+$ and $S$ a subset of vertices in $H^+$.
Assume that there exist $\alpha > 0$, $F \geq 1$, $t$ such that $\alpha \leq q_{s,v}(t) \leq F \alpha$ for every $v \in S$.
If the original graph $G$ has at least $\epsilon d |S|/2$ negative edges inside $S$ then $m \in \Omega\left( \frac{F}{\epsilon \alpha \sqrt{|S|}} \right)$ runs of $M(H^+)$ over $t$ steps and starting from $s$ will return a bad cycle with probability at least $0.99$.
\end{restatable}

Ultimately we are ofcourse interested in the behavior of random walks in $G^+$, rather than that of $M(H^+)$.
The following lemma (proven in \cref{app:mod-MC}) shows that both are closely related.

\begin{claim} \label{claim:rw-sufficient}
Assume that there exists $s$ and $m$ such that $m$ walks of $M(H^+)$ of length $\tO(1/\epsilon^3)$ and starting from $s$ result in a bad cycle with probability at least $0.99$.
Then $m$ random walks in $G^+$ of length $\tO(1/\epsilon^8)$ and starting from $s$ will also result in a bad cycle with probability at least $0.99$.
\end{claim}

\subsubsection{Correctness}

We can prove correctness of the tester (Algorithm \ref{alg:weak-balance}) by combining the ingredients from the previous section.
Rather than proving that a graph $G$ that is $\epsilon$-far from clusterable will be rejected, we will prove that if $G$ is accepted with large probability then $G$ must be close to being clusterable.
This is again similar to the proof of correctness in \cite{goldreich1999sublinear}.

\begin{claim} \label{claim:eps-good}
If Algorithm \ref{alg:weak-balance} accepts a graph $G$ with probability greater than $1/3$, then $G$ must be $2\epsilon$-close to being clusterable.
\end{claim}

To prove this, let $G$ be a graph that is accepted with probability greater than $1/3$.
We say that a vertex $s$ is \emph{good} if $\badcycle(s)$ in \cref{alg:weak-balance} returns a bad cycle with probability at most $0.1$.
Otherwise it is \emph{bad}.
Since we reject with probability less than $2/3$, and we consider $\Omega(1/\epsilon)$ starting vertices, there can be only $\epsilon N/16$ bad vertices (for the appropriate constant in the $\Omega(\cdot)$ notation).
We will show that under these circumstances we can find a valid clustering by removing less than $2 \epsilon d N$ edges.
To this end, we will iteratively separate a subset $S$ that has at most $\epsilon d |S|/2$ positive outgoing edges and at most $\epsilon d |S|/2$ negative internal edges.
We call such a subset an \emph{$\epsilon$-good cluster}.

At a given step, let $H^+$ denote the unpartitioned graph.
We wish to invoke Lemma \ref{lem:decomp}.
Call a vertex $s$ for which the lemma holds a ``useful'' vertex with respect to $H^+$.
While $|H^+| \geq \epsilon N/4$, the lemma ensures that there are $\geq \epsilon N/8$ useful vertices.
Since there are at most $\epsilon N/16$ bad vertices, this implies that there exists a vertex $s$ that is both good and useful.
By Lemma \ref{lem:decomp} there exists a subset $S$ in $H^+$ that has at most $\frac{\epsilon}{2} d |S|$ (positive) edges to the rest of $H^+$.
Moreover, by Lemma \ref{lem:bad-cycle} and Claim \ref{claim:rw-sufficient}, the set $S$ is such that if the original signed graph $G$ has at least $\epsilon d |S|$ negative edges inside $S$ then $\badcycle(s)$ will return a bad cycle with probability at least $0.99$.
However, we assumed that $s$ is a good vertex and so the latter probability can be at most $0.1$.
This implies that $G$ must have less than $\epsilon d |S|/2$ negative edges inside $S$, and hence $S$ is an $\epsilon$-good cluster.

We repeat this process until $|H^+| < \epsilon N/4$.
If $V_1,\dots,V_k$ denote the $\epsilon$-good clusters that we have cut off, then we have a partition $V = V_1 \cup \dots \cup V_k \cup H$ such that the number of positive edges between the partitions is at most
\[
d |H|
+ \sum_i |E^+(V_i,V_i^c)|
\leq \frac{\epsilon}{4} d N + \sum_i \frac{\epsilon}{2} d |V_i|
< \epsilon d N,
\]
and the number of negative edges inside the partitions is at most
\[
d |H| + \sum_i |E^-(V_i)|
\leq \frac{\epsilon}{4} d N + \sum_i \frac{\epsilon}{2} d |V_i|
< \epsilon d N.
\]
Removing these less than $2 \epsilon d N$ edges yields a valid clustering, so that $G$ must be $2 \epsilon$-close to clusterable.
This proves \cref{claim:eps-good}.

\section{Acknowledgements}
This work has benefited from discussions with Jop Bri\"et, Aristides Gionis, Oded Goldreich and Christian Sohler.
Florian Adriaens is supported by the ERC Advanced Grant REBOUND (834862), the EC H2020 RIA project SoBigData (871042), and the Wallenberg AI, Autonomous Systems andSoftware Program (WASP) funded by the Knut and Alice Wallenberg Foundation.
Simon Apers is supported in part by the Dutch Research Council (NWO) through QuantERA ERA-NET Cofund project QuantAlgo 680-91-034.

%%
%% Bibliography
%%
\bibliographystyle{alpha}
\bibliography{biblio.bib}

\appendix

\section{Technical details for dense signed graph model} \label{app:dense}

\subsection{Clusterability} \label{app:dense-cluster}

The following theorem is proven by Andersson and Engebretsen \cite{andersson2002property}.

\begin{theorem}[{\cite[Theorem 2]{andersson2002property}}] \label{thm:CSP}
Consider a constraint family $\mathcal{F} = \{f:D^\ell \to \{0,1\}\}$ over $\ell$ variables in domain $D$, and let $\Sigma$ denote the maximum number of constraints that can be simultaneously satisfied.
An $\ell$-CSP-$D$ instance $\mathcal I$ over $n$ variables with constraint family $\mathcal{F}$ is described by a collection of constraints $\{(f,x_{i_1},\dots,x_{i_\ell})\}$ with $f \in \mathcal{F}$ and $i_1,\dots,i_\ell \in [n]$.
It is possible to approximate the maximum number of satisfiable constraints $\max(\mathcal I)$ up to error $\epsilon n^k$ with probability at least $1-\delta$ using
\[
\tO\left( \frac{|\mathcal{F}| \Sigma^7 \ell^2}{\epsilon^7} \right)
\]
queries\footnote{A query to the instance $\mathcal{I}$ takes the form $Q = (f,x_{i_1},\dots,x_{i_\ell})$ and returns 1 if $Q \in \mathcal{I}$ and 0 otherwise.} and time $\exp\big( \tO\big( \frac{\Sigma^3 \ell}{\epsilon^3} \big) \big)$.
\end{theorem}
The problem of $k$-clusterability is a special instance of this problem.
Define the family $\mathcal{F} = \{f^+, f^-: [k]^2 \to \{0,1\} \}$ by $f^+(x,y) = 1$ if $x=y$ and $0$ otherwise, and $f^-(x,y) = 0$ if $x = y$ and $1$ otherwise.
Now given a signed graph $G$, we can define a collection $\mathcal{I}_G$ of constraints by adding constraint $(f^+,x,y)$ if $(x,y)$ is a positive edge, and $(f^-,x,y)$ if $(x,y)$ is a negative edge.
We can query $\mathcal{I}_G$ using a single query to the adjacency matrix of $G$.
Moreover, the $k$-frustration index of $G$ is given by
\[
|E(G)| - \max(\mathcal{I}_G),
\]
with $|E(G)|$ the number of edges in $G$.
Using that $\Sigma = 1$, $|\mathcal{F}| = 2$ and $\ell = 2$ for the family $\mathcal{F}$, it follows from Theorem \ref{thm:CSP} that we can find an $\epsilon N^2$ approximation of $\max(\mathcal{I}_G)$ using $\tO(1/\epsilon^7)$ queries and time $\exp(\tO(1/\epsilon^3))$.
In addition we can easily approximate $|E(G)|$ to additive error $\epsilon N^2$ by randomly sampling entries of the adjacency matrix of $G$.
Combining these gives an approximation algorithm for the $k$-frustration index with additive error $\epsilon N^2$, and hence a tolerant tester for $k$-clusterability.

\section{Technical details for clusterability testing in bounded-degree model} \label{app:bounded-degree}

\subsection{Modified Markov chain} \label{app:mod-MC}

In this section we describe the technical details on the modified Markov chain proposed by Goldreich and Ron \cite{goldreich1999sublinear}.
Let $G = (V,E,\sigma)$ be a signed graph and let $G^+ = (V,E^+)$ be the subgraph induced on the positive edge set.
Let $H^+$ be a subgraph of $G^+$, and let $\ell_1,\ell_2$ be integers.
The boundary $B(H^+)$ of $H^+$ consists of those vertices in $H^+$ that have an edge in $G^+$ that leaves $H^+$.
Let $\hat H^+$ be the graph obtained by appending to every boundary node $v \in B(H^+)$ an \emph{auxiliary path} of length $\ell_1$ with node set $a_{v,1},\dots,a_{v,\ell_1}$.

We will define a surjective mapping $\phi$ from random walks $W$ in $G^+$ of length $L = \ell_1 \ell_2$ to walks $\phi(W)$ in $\hat H^+$ of length $\ell_1$.
If $W = v_0,\dots,v_L$, let $i_0,\dots,i_k$ be the timesteps for which $v_{i_j} \in H^+$.
The mapping is defined essentially by contracting all length-$(<\ell_2)$ walks outside of $H^+$, and routing any length-$(\geq\ell_2)$ walk outside of $H^+$ onto an auxiliary path.
More precisely:
\begin{itemize}
\item
\emph{Contract:}
If $W$ does not perform $\ell_2$ or more consecutive steps outside of $H^+$ before it made $\ell_1$ steps (in total) in $H^+$, then
\[
\phi(W) = v_{i_0},\dots,v_{i_{\ell_1}}.
\]
\item
\emph{Contract and route:}
In the other case, let $i_r$ be the first index that precedes a walk of $\geq \ell_2$ consecutive steps outside of $H^+$.
Then
\[
\phi(W) = v_{i_0},\dots,v_{i_r},a_{v_{i_r},1},\dots,a_{v_{i_r},\ell_1-i_r}.
\]
\end{itemize}

\begin{figure}[htb]
\centering
\includegraphics[width=.9\textwidth]{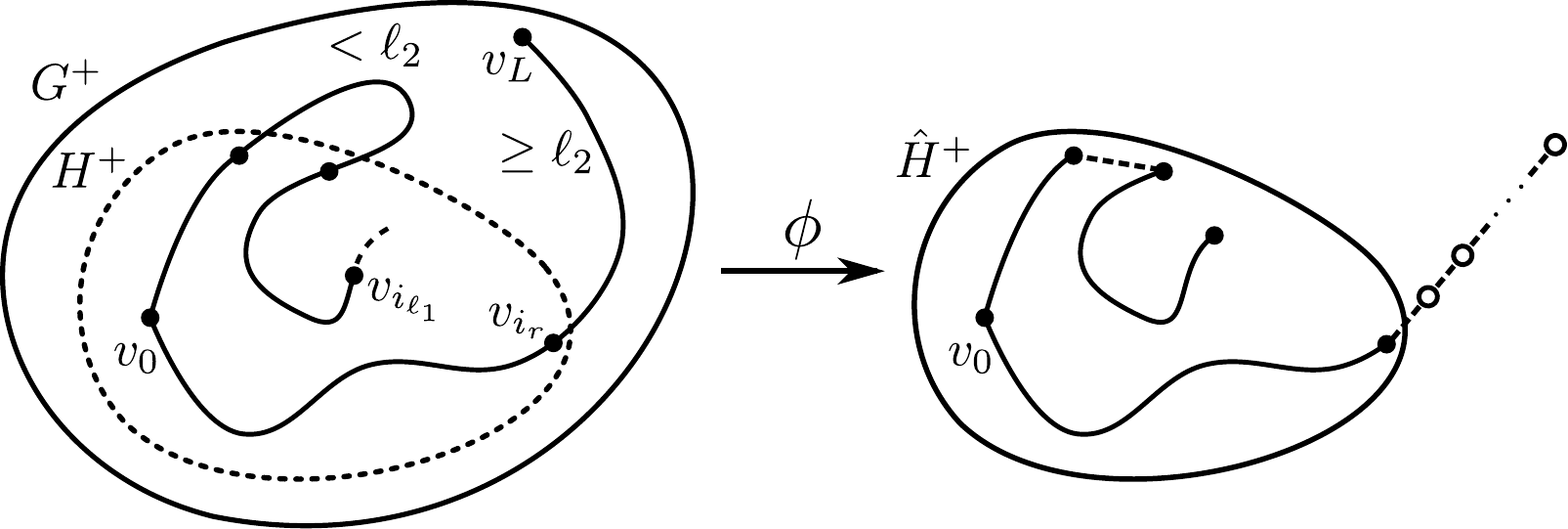}
\caption{Illustration of mapping $\phi$ from walks of length $L = \ell_1 \ell_2$ on $G^+$ to walks of length $\ell_1$ on $\hat H^+$.}
\label{fig:walk-mapping}
\end{figure}

The distribution $\Pr_G(W)$ over length-$L$ walks $W$ in $G^+$ induces a distribution $\Pr_M(U)$ over length-$\ell_1$ walks $U$ in $\hat H^+$ by setting
\[
\Pr_M(U)
= \sum_{W:\phi(W)=U} \Pr_{G^+}(W).
\]
We now define a Markov chain $M(H^+)$ on $\hat H^+$ such that length-$\ell_1$ walks $U$ of $M(H^+)$ have the same distribution $\Pr_M(U)$.
In the definition of $M(H^+)$ we use the quantity $p^H_{v,u}(t)$ for $v,u \in H^+$, which denotes the probability that a random walk from $v$ will take $t-1$ steps outside of $H^+$ and end in $u$ at the $t$-th step.
The Markov chain $M(H^+)$ is defined as follows:
\begin{itemize}
\item
For every $v,u \in H^+$: $q_{v,u} = \sum_{t=1}^{\ell_2-1} p^H_{v,u}(t)$.
\item
For every $v \in B(H^+)$:
\begin{itemize}
\item
$q_{v,a_{v,1}} = \sum_{u \in H} \sum_{t \geq \ell_1} p^H_{v,u}(t)$,
\item
for every $\ell$, $1 \leq \ell < \ell_1$, $q_{a_{v,\ell},a_{v,\ell+1}} = 1$,
\item
for every $u \in H$, $q_{a_{v,\ell_1},u} = q_{v,a_{v,1}}^{-1} \sum_{t \geq \ell_2} p^H_{v,u}(t)$.
\end{itemize}
\end{itemize}

\noindent
The following claim states that if we find a bad cycle with the modified Markov chain, then we will also find a bad cycle using the original random walk.
We say that a set of walks \emph{results in a bad cycle} if the original graph $G$ has a negative edge between two distinct vertices of the walks.

\begin{claim} \label{claim:mapping}
Assume that there exists $s$ and $m$ such that $m$ walks of $M(H^+)$ of length $\ell_1$ and starting from $s$ result in a bad cycle with probability at least $0.99$.
Then $m$ random walks in $G^+$ of length $L = \ell_1 \ell_2$ and starting from $s$ will also result in a bad cycle with probability at least $0.99$.
\end{claim}
\begin{proof}
Let $U_1$ and $U_2$ denote length-$t$ walks of $M(H^+)$ that result in a bad cycle.
If $W_1$ and $W_2$ are length-$L$ walks on $G^+$ with $\phi(W_1) = U_1$ and $\phi(W_2) = U_2$, then $W_1$ and $W_2$ will also result in a bad cycle.
Now let $I_{\mathrm{bad}}(X_1,\dots,X_m)$ denote the indicator of whether the walks $X_1,\dots,X_m$ (in $G^+$ or $M(H^+)$) form a bad cycle.
By our former remark we know that $I_{\mathrm{bad}}(W_1,\dots,W_m) \geq I_{\mathrm{bad}}(\phi(W_1),\dots,\phi(W_m))$.
We can lower bound the probability that $m$ walks in $G^+$ form a bad cycle:
\begin{align*}
\sum_{W_1,\dots,W_m} &\Pr_{G^+}(W_1) \dots \Pr_{G^+}(W_m) I_{\mathrm{bad}}(W_1,\dots,W_m) \\
&= \sum_{U_1,\dots,U_m} \left( \sum_{W_1:\phi(W_1)=U_1} \dots \sum_{W_m:\phi(W_m)=U_m} \Pr_{G^+}(W_1) \dots \Pr_{G^+}(W_m) I_{\mathrm{bad}}(W_1,\dots,W_m) \right) \\
&\geq \sum_{U_1,\dots,U_m} \left( \sum_{W_1:\phi(W_1)=U_1} \dots \sum_{W_m:\phi(W_m)=U_m} \Pr_{G^+}(W_1) \dots \Pr_{G^+}(W_m) I_{\mathrm{bad}}(U_1,\dots,U_m) \right) \\
&= \sum_{U_1,\dots,U_m} \Pr_M(U_1) \dots \Pr_M(U_m) I_{\mathrm{bad}}(U_1,\dots,U_m) \\
&\geq 0.99,
\end{align*}
which proves our claim.
\end{proof}

\subsection{Sufficient condition for bad cycle} \label{app:bad-cycle}
\badcyclelemma*
\begin{proof}
For $1 \leq i,j \leq m$, let $\eta_{i,j}$ be the random variable so that $\eta_{i,j} = 1$ if the $i$-th and $j$-th walk form a bad cycle and otherwise $\eta_{i,j} = 0$.
We will bound the probability that we do not find a bad cycle, which is $\Pr(\sum_{i<j} \eta_{i,j} = 0)$.
To this end we will use the bound
\begin{equation} \label{eq:cheb}
\Pr\left[ \sum_{i<j} \eta_{i,j} = 0 \right]
\leq \frac{\Var\left[ \sum_{i<j} \eta_{i,j} \right]}{\Exp[\sum_{i<j} \eta_{i,j}]^2}
\end{equation}
which can be derived from Chebyshev's inequality \cite[Theorem 4.3.1]{alon2004probabilistic}.

To bound $\Var\left[ \sum_{i<j} \eta_{i,j} \right]$, we define $\bar\eta_{i,j} = \eta_{i,j} - \Exp[\eta_{i,j}]$.
By \cite[Equation (20)]{goldreich1999sublinear} we can then rewrite
\[
\Var\left[ \sum_{i<j} \eta_{i,j} \right]
= \Exp\left[ \left(\sum_{i<j} \bar\eta_{i,j} \right)^2 \right]
= \binom{m}{2} \, \Exp\left[ \bar\eta_{1,2}^2 \right]
	+ 4 \binom{m}{3} \, \Exp\left[ \bar\eta_{1,2} \bar\eta_{2,3} \right].
\]
The first term is bounded by $\Exp\left[ \bar\eta_{1,2}^2 \right] \leq \Exp\left[ \eta_{1,2}^2 \right] = \Exp\left[ \eta_{1,2} \right]$.
Now we bound the second term, where we let $v_i$ denote the endpoint of the $i$-th walk:
\begin{align*}
\Exp\left[ \bar\eta_{1,2} \bar\eta_{2,3} \right]
\leq \Exp\left[ \eta_{1,2} \eta_{2,3} \right]
&= \Pr\left[ \eta_{1,2} = 1 \text{ and } \eta_{2,3} = 1 \right] \\
&= \Pr\left[ \eta_{1,2} = 1 \text{ and } (v_2,v_3) \in E^- \right] \\
&= \sum_u \Pr\left[ \eta_{1,2} = 1 \text{ and } v_2 = u \text{ and } (u,v_3) \in E^- \right] \\
&= \sum_u \Pr\left[ (u,v_3) \in E^- \right] \cdot \Pr\left[ \eta_{1,2} = 1 \text{ and } v_2 = u \right] \\
&= \sum_u \Bigg( \sum_{w:(u,w) \in E^-} \Pr\left[ v_3 = w \right] \Bigg) \cdot \Pr\left[ \eta_{1,2} = 1 \text{ and } v_2 = u \right] \\
&\leq \left( d \max_w \Pr[v_3 = w] \right) \sum_u \Pr\left[ \eta_{1,2} = 1 \text{ and } v_2 = u \right] \\
&= \left( d \max_w \Pr[v_3 = w] \right) \Exp[\eta_{1,2}] \\
&\leq d F \alpha \, \Exp[\eta_{1,2}].
\end{align*}
If we denote $\gamma = \Exp[\eta_{1,2}]$ then we get the bound $\Var\left[ \sum_{i<j} \eta_{i,j} \right]
\in O( m^2 \gamma + m^3 d F \alpha \gamma )$.
Combined with \eqref{eq:cheb} this gives the bound
\[
\Pr\left[ \sum_{i<j} \eta_{i,j} = 0 \right]
\in O\left( \frac{m^2 \gamma + m^3 d F \alpha \gamma}{m^4 \gamma^2} \right).
\]
If we let $E^-_S$ denote the set of negative edges with both endpoints inside $S$ then we can lower bound
\[
\gamma
= \sum_{(u,v) \in E^-_S} q_{s,u}(t) q_{s,v}(t)
\geq \alpha^2 |E^-_S|.
\]
Combined with the fact that $m \in \Omega\left( \frac{F}{\epsilon \alpha \sqrt{|S|}} \right)$ this gives the bound $\Pr\left[ \sum_{i<j} \eta_{i,j} = 0 \right] \in O(\epsilon/(dF^2) + 1/\sqrt{|S|})$, which is $\leq 0.01$ for the right choice of constants.
\end{proof}

\end{document}